\documentclass[12pt]{amsart}
\usepackage{amsmath, amssymb, geometry}
\usepackage[usenames,dvipsnames]{xcolor} 
\usepackage{tcolorbox} 
\usepackage{tabularx} 
\usepackage{array} 
\usepackage{colortbl} 
\tcbuselibrary{skins}
\usepackage[colorinlistoftodos]{todonotes}
\usepackage{hyperref}
\usepackage{bbm}
\graphicspath{ {./Figures/} }
\usepackage{nicefrac}
\usepackage{cleveref}

\title{Sum-Of-Squares To Approximate Knapsack}
\author{Pravesh K. Kothari and Sherry Sarkar\\Carnegie Mellon University}
\date{\today}

\theoremstyle{plain}
\newtheorem{thm}{Theorem}[section]
\newtheorem{lemma}[thm]{Lemma}
\newtheorem{cor}[thm]{Corollary}

\newtheorem{remark}{Remark}

\newcommand*\rr{\mathbb{R}}
\newcommand*\ex{\mathbb{E}}
\newcommand*\pr{\mathbb{P}}
\newcommand*\mc[1]{\mathcal{ #1 }}

\definecolor{debianred}{rgb}{0.84, 0.04, 0.33}

\theoremstyle{definition}
\newtheorem{defn}{Definition}[section]

\newcommand*\pex[1]{\Tilde{\mathbb{E}}_{ #1 }}
\newcommand*\rd{\text{Red}}
\newcommand*\pxd{\pex{D}}

\begin{document}

\begin{abstract}
These notes give a self-contained exposition of Karlin, Mathieu and Nguyen's~\cite{KMN} tight estimate of the integrality gap of the sum-of-squares semidefinite program for solving the knapsack problem. They are based on a sequence of three lectures in CMU course on Advanced Approximation Algorithms in Fall'21 that used the KMN result to introduce the Sum-of-Squares method for algorithm design. The treatment in these notes uses the pseudo-distribution view of solutions to the sum-of-squares SDPs and only rely on a few basic, reusable results about pseudo-distributions. 
\end{abstract}

\maketitle

\section{Approximating Knapsack}\label{intro}

\subsection{The Knapsack Problem} We're given $n$ items, each item $i$ accompanied with a capacity $c_i \geq 0$ and a value $v_i \geq 0$. We are allowed to pick a set of items with total capacity at most $C$, and the goal is to maximize the total value of the set of items we pick. In the following, we will use OPT to denote the maximum total value that we can collect by picking any subset of items that satisfy the capacity constraint. OPT can then be computed using the following integer program. 
\begin{align*}
    \max \quad &\sum_{i= 1}^n v_i x_i \\
    \sum_{i = 1}^n &c_i x_i \leq C \\
    x_i &\in \{0, 1\} \quad \forall i \in [n]. 
\end{align*}
The knapsack problem (and, in particular, the above integer program) is NP hard to solve exactly. 

By relaxing $x_i \in [0,1]$ from $x_i \in \{0,1\}$, we obtain a linear program. While this basic linear program does not give a good approximation ratio, it is easy to derive the following upper bound on the LP value that is meaningful if the maximum value of any item is not too large compared to the true optimum. 

\begin{lemma}
Let LP-VAL $= \max \sum_{i = 1}^n v_i x_i $ as $x$ varies over the set $\{| \sum_{i \leq n} c_i x_i \leq C, \text{ and } 0 \leq x_i \leq 1 \forall i\}$. Then,
\[ LP-VAL \leq OPT + \max_{i \in [n]} v_i. \]
\end{lemma}

In this note, we will analyze the integrality gap of a sequence of tighter semidefinite programs that can be mechanically built from the above integer program via the sum-of-squares method. Analyzing these relaxations is an excuse for us to introduce some basic methodology of dealing with the sum-of-squares SDP. In particular, instead of the traditional vector programming view of SoS SDPs, we will use the \emph{pseudo-distribution} view in these notes. This view is helpful in inspiring both new SDP relaxations and their analyses for several interesting problems in the past decade. While these applications are beyond the scope of these notes, we refer the reader to the notes from a CMU course~\cite{notes-lec, videos-lec} on this topic from Fall 2020 and a monograph~\cite{FKP19} that focuses on the connections of this method to proof complexity. 

These notes are organized as follows. In Section~\ref{pseudo}, we will start with a different inefficient (in fact, even writing solutions down for this relaxation would be inefficient) program to exactly compute solutions for knapsack. We will then consider a natural graded relaxation of this program whose solutions can be naturally viewed as relaxations of probability distributions over the hypercube $\{-1,1\}^n$ that we will define as pseudo-distributions. 

\section{Another Inefficient Program}\label{pseudo}
The program in Section~\ref{intro} asks for an integer assignment to the vector $x$ that satisfies the knapsack capacity constraint while maximizing the total value collected. 

We now consider a new program, that we will call, the \emph{distribution program}, that searches for probability distributions over integer solutions -- i.e., a probability distribution $\mu$ over $x \in \{-1,1\}^n$ maximizing the expected total value of items collected over all $\mu$ such that that every point in the support of $\mu$ satisfies the knapsack capacity constraint.

Pseudo-distributions are relaxations of probability distributions. While they can be defined on the solution space of any system of polynomial equalities and inequalities, in these notes, we will restrict to the hypercube $\{-1,1\}^n$. Note that this is the solution space to the system $\{x \mid x_i^2 = 1\}$. 


Knapsack, along with many other combinatorial problems (max-cut, etc.), can be reduced to finding a maximum of a polynomial of a hypercube. In particular, we are interested in the generalization (P)
\begin{gather*}
    \max \quad p(x) \\
    \text{s.t} \quad x \in \{-1, 1\}^n
\end{gather*}
where $p(x)$ is some polynomial. Consider the related problem (Q)
\begin{gather*}
    \max \quad \ex_{\mu}(p(x)) \\
    \text{s.t $\mu$ is a distribution over $\{-1, 1\}^n$ }.
\end{gather*}
Clearly (Q) is a relaxation of (P). However, notice that for any distribution $\mu$, 
\[ \ex_{\mu}(p(x)) \leq \max_{x \in \text{supp}(\mu)} p(x) \]
and hence a distribution with singular support maximizes $p(x)$. So, to solve (P) it is equivalent to find a probability distribution with high expectation instead.  

The classic way to think about a probability distribution is in terms of mass on elements in $\{-1, 1\}^n$. However, an equivalent way to think about probability distributions, is by the moments i.e., the expectation of monomials $X_S = \prod_{i \in S} x_i$. Defining a distribution via its support clearly defines the expectation of the monomials. The converse is true as well; namely, we may look at the indicator functions 
\[ \chi_y(x) := \prod_{i \mid x_i = 1} \left(\frac{1 + x_i}{2} \right) \prod_{i \mid x_i = -1} \left(\frac{1 - x_i}{2} \right) \]
and observe that for a fixed $y \in \{-1, 1\}^n$, 
\[ \pr_{\mu}(x = y) = \ex_{\mu}(\chi_y(x)). \]
Going back to our program (Q), once we have defined $\mu$ in terms of the monomials, we have a liner objective function. It is easy to constrain $\sum_{x \in \{-1, 1\}^n} \mu(x) = 1$ -- this is a single linear constraint. However, making sure all the moments are non-negative as they are in a probability distribution would require exponential constraints. Hence, we settle for a ``pseudo-distribution", or even better, a pseudo-distribution where some of the smaller moments are positive.

\begin{defn}[Pseudo-Distribution]
A pseudo-distribution over $\{-1, 1\}^n$ is a function $D: \{-1, 1\}^n \rightarrow \mathbb{R}$ such that $\sum_{x \in \{-1, 1\}^n} D(x) = 1$. 
\end{defn}

Unlike probability distributions, pseudo-distributions may take on negative values over the support. We define various other probability terms for pseudo-distributions. 

\begin{defn}[Pseudo-Expectation ]
The pseudo-expectation of a function $f: \{-1, 1\}^n \rightarrow \rr$ with respect to pseudo-distribution $D$ over $\{-1, 1\}^n$ is
\[ \pex{D}(f) := \sum_{x \in \{-1, 1\}^n} D(x) \cdot f(x). \]
\end{defn}

The pseudo-expectation of a probability distribution corresponds to the expectation of a probability distribution. 

\begin{lemma}[Linearity]
For two functions $f, g: \{-1, 1\}^n \rightarrow \mathbb{R}$ and a pseudo-distribution $D$ over $\{-1, 1\}^n$, 
\[ \pex{D}(f + g) = \pex{D}(f) + \pex{D}(g).\]
\end{lemma}

Not only can we add two polynomials in a vector space, we can multiply two polynomials. Multiplication is a bi-linear operation. Indeed, we can check $\pex{D}(u \cdot v)$ is a bilinear form (in a vector space, this is analogous to a dot-product). This means that given a pseudo-distribution $D$, there exists a matrix $M$ such that 
\[ \pex{D}(u \cdot v) = u^T M v. \]
What does $M$ look like? We have $M_t(S, T) = \pex{D}(X_{S \triangle T}). $

\begin{defn}
A pseudo-distribution $D$ over $\{-1, 1\}^n$ is of degree $d$ if for all polynomials $f$ of degree at most $\nicefrac{d}{2}$, we have 
\[ \pex{D}(f^2) \geq 0. \]
\end{defn}

Hence, if $D$ is a pseudo-distribution of degree $d$, we know that for all polynomials $u$ of degree less than or equal to $\nicefrac{d}{2}$, $u^T M u \geq 0$. In particular, $M$ restricted to the vector space of polynomials of degree less than or equal to $\nicefrac{d}{2}$ is positive semi-definite. Note that $M$ is a $\binom{n}{\nicefrac{d}{2}} \times \binom{n}{\nicefrac{d}{2}}$ dimension matrix.

\begin{thm}[Cauchy-Schwartz]
Let $D$ be a pseudo-distribution of degree $2d$ on $\{-1, 1\}^n$. Then, for all polynomials of degree at most $d$, 
\[ \pex{D}(f \cdot g) \leq \sqrt{\pex{D}(f^2)} \cdot \sqrt{\pex{D}(g^2)}. \]
\begin{proof}
The proof relies on observing that a natural argument for establishing the usual Cauchy-Schwarz inequality extends, in general, to any positive semidefinite bilinear forms. We include it in full for completeness. 

Suppose first that $\pex{D}(f^2)=0$. Since $D$ is a degree $2d$ pseudo-distribution, for every $C > 0$, 
\[ \pex{D}\left(Cf - \frac{g}{C}\right)^2 \geq 0.\]
Expanding out, this yields that 
\[ 2\pex{D}(fg) \leq\pex{D}(C^2 f^2) + \pex{D}\left(\frac{g^2}{C^2} \right) = \frac{1}{C^2} \pex{D}(g^2). \]
Letting $C \rightarrow \infty$ gives $\pex{D}[fg] \leq 0$. A similar argument starting with $\pex{D}(Cf + g/C)^2 \geq 0$ yields that $-\pex{D}(fg) \leq 0$. Together, we conclude that $\pex{D}(fg)=0$ completing the proof. 

Let's now assume that $\pex{D}(f^2), \pex{D}(g^2) > 0$. In this case, let 
\[ \bar{f} = \frac{f}{\sqrt{\pex{D}(f^2)}}, \quad \bar{g} = \frac{g}{\sqrt{\pex{D}(g^2)}}. \] 
Then, we have $ \pex{D}(\bar{f}+\bar{g})^2\geq 0$ and thus, $\pex{D}(\bar{f}\bar{g})\leq \frac{1}{2} \pex{D}(\bar{f}^2 + \bar{g}^2) = 1$. Rearranging yields that $\pex{D}(fg)\leq \sqrt{\pex{D}(f^2)} \sqrt{\pex{D}g^2}$ as desired. A symmetric argument starting with $ \pex{D}(\bar{f}-\bar{g})^2\geq 0$ yields that $-\pex{D}(fg)\leq \sqrt{\pex{D}(f^2)} \sqrt{\pex{D}g^2}$. This completes the proof.

\end{proof}
\end{thm}

Note that the indicator function 
\[ \chi_y(x) := \prod_{i \mid x_i = 1} \left(\frac{1 + x_i}{2} \right) \prod_{i \mid x_i = -1} \left(\frac{1 - x_i}{2} \right) \]
is a degree $n$ polynomial. 

\begin{cor}
A pseudo-distribution of degree $2n$ is simply a probability distribution. 
\end{cor} 

\subsection{A Brief Introduction to SDPs}

Classical linear programming is often of the form 
\begin{align*}
    \max \quad &c^Tx \\
    Ax &\leq b \\
    x &\geq 0,
\end{align*}
We are maximizing $c^Tx$ over the cone of non-negative vectors, subject constraints $Ax \leq b$. This is a special instance of conic programming. Similarly, semi-definite programs are trying to maximize a function over the cone of positive semi-definite matrices. 

\begin{align*}
    \max \quad C &\bullet X \\
    \mc{A}(x) &\leq b \\
    x &\succeq 0,
\end{align*}
where $C \in \mathbb{S}^n$, $\mc{A}: \mathbb{S}^n \rightarrow \mathbb{R}^m$, and $\bullet$ refers to the Frobenius dot product. 

As mentioned in the beginning of this survey, we often think about distributions as defined by their moments. We introduce the natural generalisation to pseudo-distributions. 
\begin{defn}[Pseudo-moments]
The degree $t$ pseudo-moments of a pseudo-distribution $D$ are the set of numbers 
\[ \{ \pex{D}(X_S) \mid |S| \leq t \} \]
i.e., the pseudo-expectations of monomials of degree at most $t$.
\end{defn}

Now we have all the tools to understand how to guarantee a pseudo-distribution has degree $2t$ via semi-definite programming. 

\begin{lemma}
The set $\{ \gamma_S \mid |S| \leq 2t \}$, where $\gamma_{\emptyset} = 1$, are the pseudo-moments of a degree $2t$ pseudo-distribution $D$ on $\{-1, 1\}^n$ if and only if the matrix $M_t$ defined by 
\[ M_t(S,T) = \pex{D}(X_{S \triangle T}) \]
is positive semi-definite. 
\end{lemma}
\begin{proof}
$\impliedby$: We define the pseudo-moments of $D$ as follows - 
\[ \pex{D}(X_S) = \begin{cases} \gamma_S & |S| \leq 2t \\ 0 & |S| > 2t \end{cases}. \]
We first prove this is a pseudo-distribution. The only condition we need to check is that $\sum_{x \in \{-1, 1\}^n} D(x) = 1$. Indeed, $\sum_{x \in \{-1, 1\}^n} D(x)$ is just the moment of the empty set which we enforce as equal to 0. Next, we have to prove that $D$ is of degree at least $2t$. Take a function $f$ of degree $t$. We express $f$ as 
\[ f = \sum_{S \subseteq [n]} \hat{f}_S X_S = \sum_{S \subseteq [n] , |S| \leq t} \hat{f}_S X_S. \]
In terms of the moments basis, the vector representation of $f$ (which belongs to the space of polynomials of degree at most $t$) is clearly $\hat{f} = (\hat{f}_S)_{S \subseteq [n], |S| \leq t}$. Recall $\pex{D}(f^2) = \hat{f}^T M_t \hat{f} \geq 0$, since $M_t$ is positive semi-definite. 

$\implies$: We prove $M_t$ is positive semi-definite. This is almost identical to the last direction; we need only show that $\hat{f}^T M_t \hat{f} \geq 0$ for all vectors $\hat{f}$ is the space of polynomials of degree at most $t$. Since $\hat{f}$ corresponds to a polynomial of degree at most $t$, by definition of $D$ being degree at least $2t$, we know that $\hat{f}^T M_t \hat{f} = \pex{D}(f^2) \geq 0$. 
\end{proof}

\subsection{Satisfying a Constraint}
We showed in the last sub-section how an $O(n^t) \times O(n^t)$ PSD matrices correspond to pseudo-distributions of degree $2t$. So, finding a pseudo-distribution which maximizes a certain objective function reduces to semi-definite programming. What if we want to impose additional restrictions on the pseudo-distribution?

\begin{defn}[Equality Constraints on a Pseudo-distribution]

A degree $2t$ pseudo-distribution $D$ on $\{-1, 1\}^n$ satisfies a constraint $q(x) = 0$ if
\[ \pex{D} (q \cdot X_S) = 0 \]
for all $|S|$ such that $|S| + \text{deg}(q) \leq 2t$. This is equivalent to the condition 
\[ \pex{D} (q \cdot p) = 0 \]
for all polynomials $p$ such that $\text{deg}(p) + \text{deg}(q) \leq 2t$.
\end{defn}

One may interpret the condition $\pex{D}(q \cdot p) = 0$ for all polynomials $p(x)$ such that $\text{deg}(p) \leq 2t - q$ as $\langle q, p \rangle_{D} = 0$ for all vectors $p$ in a subspace of dimension $2t - q$. Indeed, in Euclidean space a vector $v$ satisfying the condition $\langle v , w \rangle = 0$ for all vectors $w$, is equivalent to $v = 0$. 

To encode inequality constraints, a bit more work has to be done.

\begin{defn}[Inequality Constraints on a Pseudo-distribution]
A degree $2t$ pseudo-distribution $D$ on $\{-1, 1\}^n$ satisfies a constraint $q(x) \geq 0$ if the matrix $M_q$ defined as 
\[ (M_q)_{S, T} := \pxd(q(x) \cdot X_{S \Delta T})\]
for $S, T$ of sizes at most $t - \text{deg}(q)/2$ is positive semi-definite. In other words, 
\[ \pxd(q(x) \cdot f^2) \geq 0 \]
for all $f$ of degree at most $t - \frac{\text{deg}(q)}{2}$. 
\end{defn}

This condition is inspired from the following equivalent definition of a vector $v$ in Euclidean space being non-negative: $v \geq 0$ if and only if $\langle v, w \rangle \geq 0$ for all $w \geq 0$.  

\begin{remark}
\label{rem:constraints}
Here, we study what it means for a probability distribution (which has infinite degree) $\mu$ on $\{-1, 1\}^n$ to satisfy a constraint $q(x)$. Take the indicator function $\chi_y(x)$ which takes on value 1 at $y$. Then, 
\begin{align*}
    0 \leq \ex_\mu{\chi_y^2(x) \cdot q(x)} &= \ex_\mu{\chi_y(x) \cdot q(x)} \\
    &= \sum_{x\in \{-1,1\}^n} \mu(x) \chi_y(x) \cdot q(x) \\
    &= \mu(y) \cdot q(y)
\end{align*}
Hence, if $y$ is in the support of $p$, we may conclude $q(y) \geq 0$. Every point in the support of $p$ satisfies $q(x) \geq 0$.
\end{remark}

Let's motivate why this definition of ``constrained" is natural. When considering vectors $v \in \rr^n$, the constraint $v$ non-negative, i.e., $v \geq 0$, is equivalent to 
\[ \langle v, u \rangle \geq 0 \]
for all $u \geq 0$. In our setting, to show $q(x)$ is non-negative we would ideally like to show that, with respect to our bi-linear form, $\langle q(x), p(x) \rangle_{D}$ is at least 0 for all non-negative polynomials $p$. However, it is difficult to reason about non-negative polynomials; hence we settle for the set of polynomials we can easily verify are non-negative -- squares.

\begin{cor}
For any $q_1(x), \hdots q_m(x)$ and $p(x)$ all of degree at most $t$, finding a degree $2t$ pseudo-distribution satisfying the constraints $q_i(x) \geq 0$ which maximizes $p(x)$ is reducible to a semi-definite program of size $\text{poly}(m, n^{O(t)})$.
\end{cor}
\begin{proof}
We explicitly write out the SDP. 
\begin{gather*}
    \max_{D} \quad M \cdot p \\
    M_{\emptyset, \emptyset} = 1 \\
    M \succeq 0 \\
    M_{q_i} \geq 0 \quad \text{for all $i \in [m]$.} 
\end{gather*}
\end{proof}

Next, we study restrictions of pseudo-distributions to a subset of variables. 
\begin{defn}
A restriction of pseudo-distribution $D$ to variables $S$ is a pseudo-distribution $D'$ over the variables in $S$ such that 
\[ D'(y) = \sum_{x \mid x_{|S} = y} D(x). \]
\end{defn}

\begin{lemma}[Local Distributions]
Suppose $D$ is a local distribution on $\{-1, 1\}^n$ of degree at least $2t$. Consider the restriction of $D$ to any set of variables $S$, where $|S| \leq t$. Then, there is a probability distribution $\mu$ on $\{-1, 1\}^S$ such that for all $T \subseteq S$, $\pex{D}(X_T) = \ex_{\mu}(X_T)$. 
\end{lemma}

At this point, we instead talk about pseudo-distributions over $\{0, 1\}^n$, since it is more applicable to the Knapsack problem we are considering.

\begin{lemma}
Suppose for all $x \in \{0, 1\}^n$, $\sum_{i = 1}^n c_i x_i \leq C$ implies $\sum_{i = 1}^n x_i \leq k$. Then, for every pseudo-distribution of degree at least $2k + 2$ satisfying $\sum_{i = 1}^n c_i x_i \leq C$, it holds that $\pex{D}(X_S) = 0$ for all $S$ such that $k + 1 \leq |S| \leq 2k + 2$. 
\end{lemma}
\begin{proof}
\textit{Case 1.} We start with the case $|S| = k + 1$. Since $D$ is a pseudo-distribution of degree at least $2k + 2$ and $S$ has size $k + 1$, we can consider the local distribution $\mu$ restricted to variables $S$. Recall that $\mu$ and $D$ agree on moments $T$ where $T$ is a subset of $S$. We want to show that $\mu$ also satisfies the constraint $\sum_{i = 1}^n c_i x_i \leq C$ -- if we had this, then by \Cref{rem:constraints} every point $x$ in the support of $\mu$ satisfies the inequality $\sum_{i = 1}^n c_i x_i \leq C$, which in turn by the assumptions of the theorem, imply that $\sum_{i = 1}^n  x_i \leq k$. Therefore, $\pex{D}(X_S) = \mathbb{E}_{\mu}(X_S) = 0$. 

So it remains to show that $\mu$ satisfies the constraint $\sum_{i = 1}^n c_i x_i \leq C$. Recall $\mu$ is supported on $\{0, 1\}^S$; so for $\mu$ to satisfy the constraint, we would need 
\[ \pex{\mu}\left( \left( C - \sum_{i \in S} c_i x_i \right) \cdot f^2 \right) \geq 0 \]
for $f:\{0, 1\}^S \rightarrow \mathbb{R}$ of degree at most $\frac{2k + 1}{2}$ (effectively, at most $k$). Since $\mu$ and $D$ agree on all moments which are subsets of $S$, we have 
\[ \pex{\mu}\left( \left( C - \sum_{i \in S} c_i x_i \right) \cdot f^2 \right) = \pex{D}\left( \left( C - \sum_{i \in S} c_i x_i \right) \cdot f^2 \right) \]

So, at this point, if we show
\[ \pxd\left(\left(C - \sum_{i \in S} c_i x_i \right) f^2 \right) \geq 0 \]
for any $f$ of degree at most $k$, then we are done.

From the definition of satisfying a constraint, we have
\[ \pex{D}\left(\left(C - \sum_{i=1}^n c_i x_i \right) f^2 \right) \geq 0. \]
Expanding the left hand side, we get
\begin{align*}
    \pex{D}\left( \left(C - \sum_{i=1}^n c_i x_i \right) f^2 \right) &= \pex{D}\left( \left(C - \sum_{i \in S}^n c_i x_i \right) f^2 \right) - \pex{D}\left( \left(\sum_{i\not \in S} c_i x_i \right) f^2 \right) 
\end{align*}
so it suffices to show 
\[ \pex{D}\left( \left(\sum_{i\not \in S} c_i x_i \right) f^2 \right) \geq 0. \]
Note that 
\begin{align*}
    \pex{D}\left(f^2 \left(\sum_{i\not \in S} c_i x_i \right) \right) &= \sum_{i\not \in S} c_i \pex{D}(f^2 x_i)
\end{align*}

Note that $\pex{D}(f^2 x_i) = \pex{D}(f^2 x_i^2) = \pex{D}( (f x_i)^2 )$, and since $fx_i$ is a polynomial of degree at most $k + 1$, by the degree of our pseudo-distribution, $\pex{D}(f^2 x_i^2) \geq 0$.

And since $c_i \geq 0$, we may conclude
\[ \pex{D}\left(\left(\sum_{i\not \in S} c_i x_i \right) f^2 \right) \geq 0 \]
as desired.

\textit{Case 2.} Take the case where $|S| > k + 1$. Let $T \subseteq S$ be a set of size $k + 1$. Then, 
\[ X_S = X_T \cdot X_{S \setminus T}. \]
Note that $|S \setminus T| \leq 2k + 2 - (k + 1) = k + 1$. Hence, by Cauchy-Schwartz, 
\[ \pxd(X_S) \leq \sqrt{\pxd(X_T^2)} \cdot \sqrt{\pxd(X_{S \setminus T}^2)} = 0.  \]
\end{proof}

\begin{lemma}[Global Distributions]
Suppose $D$ is a pseudo-distribution of degree at least $2k + 2$ on $\{0, 1\}^n$ such that for all $S \subseteq [n]$ such that $k + 1 \leq |S| \leq 2k + 2$, we have $\pex{D}(X_S) = 0$. Then, there exists a probability distribution $\mu$ over $\{0, 1\}^n$ such that for all $S \subseteq [n]$ such that $k + 1 \leq |S| \leq 2k + 2$, $\ex_\mu{X_S} = \pex{D}(X_S)$. 
\end{lemma}
\begin{proof}
We define $\mu$ as follows: 
\[ 
\pex{\mu}(X_S) = \begin{cases} \pex{D}(X_S) & |S| \leq 2k + 2 \\
0 & |S| > 2k + 2. 
\end{cases}
\]
To show $\mu$ is a probability distribution over $\{0, 1\}^n$, it suffices to show 
\[ \pex{\mu}(f^2) \geq 0 \]
for all polynomials $f: \{0, 1\}^n \rightarrow \rr$. To see this, we write $f$ as 
\[ f = \sum_{S \mid |S| \leq k + 1} \hat{f}_S X_S + \sum_{S \mid |S| > k + 1} \hat{f}_S X_S . \]
Call the former term $f_{sm}$ and the latter term $f_{lg}$. Then, 
\[ f^2 = f_{sm}^2 + 2 f_{sm}f_{lg} + f_{lg}^2.  \]
And so, 
\begin{align*}
    \pex{\mu}(f^2) &= \pex{\mu}(f_{sm}^2) + 2 \pex{\mu}(f_{sm} f_{lg} ) +  \pex{\mu}(f_{lg}^2) \\
    &\geq 0 
\end{align*}
since the first term is of degree at most $2k + 2$, and the latter terms are of degree larger than $2k + 2$.
\end{proof}

\section{An Approximation Algorithm for Knapsack}

\subsection{The Theorem and Proof}
\begin{thm}
\label{main-thm}
Consider a Knapsack instance with $n$ items with costs $c_1, \hdots, c_n$ and values $v_1, \hdots, v_n$. Then, for every pseudo-distribution on $x_1, \hdots, x_n \in \{0, 1\}$ of degree at least $2t$ satisfying $\sum_{i = 1}^n c_i x_i \leq C$, 
\[ \pex{D}\left( \sum_{i = 1}^n v_i x_i  \right) \leq \left( 1 + \frac{1}{t - 1} \right) OPT. \]
In particular, the integrality gap of the degree $2t$ SoS SDP is at most $\left( 1 + \frac{1}{t - 1} \right)$. 
\end{thm}
\begin{proof}[Analog proof for probability distributions]
To give motivation for how we approach the problem with pseudo-distributions, let's first prove this for probability distributions $\mu$ satisfying the constraint $\sum_{i = 1}^n c_i x_i \leq C$. We want to show $\ex_{\mu}(\sum_{i = 1}^n v_i x_i)$ is not too large. Consider the expectation of $\sum_{i = 1}^n v_i x_i$ when $U \subseteq S$ is chosen; in particular, we have
\[ \ex_{\mu}\left(\sum_{i = 1}^n v_i x_i \right) = \sum_{U \subseteq S} \ex_{\mu}\left( \sum_{i = 1}^n v_i x_i \cdot f_{S, U}\right).\]
Next, we decompose the sum $\sum_{i = 1}^n$ into terms in $S$ and terms not in $S$. We notice that $\ex_{\mu}(f_{S, U}) = \pr(x_{|S} = U)$. Therefore, 
\begin{align*}
    \ex_{\mu}\left(\sum_{i = 1}^n v_i x_i \right) &= \sum_{U \subseteq S} \ex_{\mu}\left( \sum_{i = 1}^n v_i x_i \cdot f_{S, U}\right) \\
    &= \sum_{U \subseteq S} \left( \ex_{\mu}\left( \sum_{i \in S} v_i x_i \cdot f_{S, U} \right) + \ex_{\mu}\left( \sum_{i \not \in S} v_i x_i \cdot f_{S, U} \right) \right) \\
    &= \sum_{U \subseteq S} \left( \left( \sum_{i \in U}v_i \right) \ex_{\mu}(\cdot f_{S, U}) + \ex_{\mu}\left( \sum_{i \not \in S} v_i x_i \cdot f_{S, U} \right) \right)
\end{align*}
For the latter term, we recall that we defined $y_i^U$ to satisfy
\[ \ex_{\mu}(x_i \cdot f_{S, U}) = \ex_{\mu}(f_{S, U}) \cdot y_i^U. \]
Hence, 
\begin{align*}
    \ex_{\mu}\left(\sum_{i = 1}^n v_i x_i \right) &= \sum_{U \subseteq S} \left( \left( \sum_{i \in U}v_i \right) \ex_{\mu}(\cdot f_{S, U}) + \ex_{\mu}\left( \sum_{i \not \in S} v_i x_i \cdot f_{S, U} \right) \right) \\
    &= \sum_{U \subseteq S} \left( \left( \sum_{i \in U}v_i \right) \ex_{\mu}(\cdot f_{S, U}) + \sum_{i \not \in S} v_i y_i^U \cdot \ex_\mu(f_{S, U}) \right) \\
    &= \sum_{U \subseteq S} \ex_\mu(f_{S, U}) \cdot \left( \sum_{i \in U}v_i + \sum_{i \in S} v_i y_i^U \right)
\end{align*}
At this point, we recall that $y_i^U$ is between 0 and 1, and we have $\sum_{i \not \in S} y_i^U \leq C - \sum_{i \in U} c_i$. Therefore, by the linear programming bound on Knapsack, we have 
\[ \sum_{i \in S} v_i y_i^U \leq OPT\text{ restricted to items $U$ chosen} + \frac{OPT}{t - 1}. \]
So, in total, we get
\begin{align*}
    \ex_{\mu}\left(\sum_{i = 1}^n v_i x_i \right) &= \sum_{U \subseteq S} \ex_\mu(f_{S, U}) \cdot \left( \sum_{i \in U}v_i + \sum_{i \in S} v_i y_i^U \right) \\
    &\leq \sum_{U \subseteq S} \ex_\mu(f_{S, U}) \left( \sum_{i \in U}v_i + OPT(C - \sum_{i \in U} c_i) + \frac{OPT}{t - 1} \right) \\
    &= \sum_{U \subseteq S} \ex_\mu(f_{S, U}) OPT + \frac{OPT}{t - 1} \\
    &= \left( 1 + \frac{1}{t - 1} \right) OPT
\end{align*}
as desired. 
\end{proof}

Let $S := \{ i \mid v_i \geq \frac{OPT}{t - 1} \}$. Then, we clearly have that $\sum_{i \in S} x_i \geq t$ implies $\sum_{i = 1}^n c_i x_i > C$. Hence, we have the implication $\sum_{i = 1}^n c_i x_i \leq C$ implies $\sum_{i \in S} x_i \leq t-1$.

\begin{remark}\label{looks-like-mu}
We have that for any pseudo-distribution of degree at least $2t$ satisfying $\sum_{i = 1}^n c_i x_i \leq C$, it must hold that, for all $T \subseteq S$ where $t \leq |T| \leq 2t$, $\pex{D}(X_T) = 0$. In particular, we know that $D$ restricted to $S$ has a probability distribution $\mu$ which agrees with it up to $2t$ moments. Note that $S$ may be of size much larger than $2t$, so this is not just the local distribution property. 
\end{remark}

Let us call $D'$ as the distribution $D$ restricted to $S$. We define $f_{S, U}: \{0, 1\}^S \rightarrow \rr$ as 
\[ f_{S, U}(x) = \begin{cases} 1 & \text{if $x_i = 1$ for all $i \in U$ and $x_i = 0$ for all $i \in S \setminus U$.} \\
0 & \text{otherwise}
\end{cases}.
\]
In other words, $f_{S, U}$ is an indicator function over $ \{0, 1\}^S$. 

\begin{lemma} 
\label{sum-of-indicators}
For all $x \in \{0, 1\}^S$, 
\[ \sum_{U \subseteq S} f_{S, U}(x) = 1. \]
\end{lemma}
\begin{proof}
Trivial.
\end{proof}

\begin{lemma}
For all $T \subseteq S$ such that $|T| \leq 2t - 1$ and for all $i \not \in T$,
\[ \pex{D'}(X_T \cdot x_i) = 0. \]
\end{lemma}
\begin{proof}
Let $T_1 \subseteq T$ be a subset of size $t$. Then, by Cauchy-Schwartz
\begin{align*}
    \pex{D'}(X_T \cdot x_i) = \pex{D'}(X_{T_1} \cdot X_{T \setminus T_1} \cdot x_i) \\
    &\leq \sqrt{\pex{D'}(X_{T_1}^2)} \cdot \sqrt{\pex{D'}(X_{T \setminus T_1}^2 \cdot x_i^2)} \\
    &\leq 0 \cdot \sqrt{\pex{D'}(X_{T \setminus T_1}^2 \cdot x_i^2)} \\
    &= 0.
\end{align*}
Note we can apply Cauchy-Schwartz because $D'$ is of degree $2t$, and $\text{deg}(X_{T_1}) = t$ and $\text{deg}(X_{T \setminus T_1} \cdot x_i) \leq t$. 
\end{proof}

At this point, we want to reason about the behavior of $\pex{D}$ on $f$. However, $\pex{D}$ really only has well-defined behavior for polynomials of degree less than or equal to $2d$. The follow operation will help us relate the pseudo-expectation of high degree polynomial to a low-degree one.
\begin{defn}[Reduction of $f$]
For any polynomial $f = \sum_{T \subseteq S} \hat{f}_T X_T$ over $\{0, 1\}^S$, we define $\rd(f)$ as just the low-degree terms of $f$, i.e., 
\[ \rd(f) := \sum_{T \subseteq S \mid |T| < t} \hat{f}_T X_T. \]
\end{defn}

Recall that $f$ is an element of the vector space of functions (or really, polynomials) from $\{0, 1\}^S$ to $\rr$, and hence, $\rd(f)$ can be seen as simply projecting $f$ onto the subspace of polynomials of low-degree. In other words, $\rd(f)$ is a liner operator over polynomials. 

\begin{lemma} For all $x \in \{0, 1\}^S$, we have,
\[ \sum_{U \subseteq S} \rd(f_{S, U})(x) = 1. \]
\end{lemma}
\begin{proof}
Recall that \Cref{sum-of-indicators} asserts $\sum_{U \subseteq S} f_{S, U}(x) = 1$. Since $\rd$ is a linear operator on polynomials, we have 
\[ \rd \left(\sum_{U \subseteq S} f_{S, U}(x)\right) = \sum_{U \subseteq S} \rd(f_{S, U}(x)). \]
Also, clearly, 
\[ \rd(1) = 1. \]
So, we have 
\[ 1 = \rd \left(\sum_{U \subseteq S} f_{S, U}(x)\right) = \sum_{U \subseteq S} \rd(f_{S, U}(x)). \]
\end{proof}

\begin{lemma} 
\label{red-x-i}
For all $i \in S$
\[ \sum_{U \subseteq S} \rd(f_{S,U}) \cdot x_i = x_i. \]
\end{lemma}
\begin{proof}
\begin{align*}
    x_i &= 1 \cdot x_i \\
    &= \left( \sum_{U \subseteq S} \rd(f_{S, U}) \right) \cdot x_i \\
    &= \sum_{U \subseteq S} \rd(f_{S, U}) \cdot x_i.
\end{align*}
\end{proof}

\begin{defn}[``Conditional Expectation"]
We define, for a given $i \in [n]$
\[ y_i^U := \frac{\pex{D}(x_i \cdot \rd(f_{S, U}))}{\pex{D}(\rd(f_{S, U}))}. \]
\end{defn}

We state several key facts we'll need about $y_i^U$, and then use these facts to finally present a proof of the main theorem. 

\begin{lemma}
\label{y}
For all $i \in [n]$,
\begin{enumerate}
    \item[(a)] 
    \[ 0 \leq y_i^U \leq 1 \] 
    and 
    \item[(b)] \[ \sum_{i \not \in S}y_i^U c_i \leq C - \sum_{i \in U} c_i.\]
\end{enumerate}
\end{lemma}

\begin{proof}[Real proof of \Cref{main-thm}]
Take a pseudo-distribution $D$ of degree $2t$ which satisfies the constraint $\sum_{i = 1}^n c_i x_i \leq C$. We want to analyze $\pex{D}\left( \sum_{i = 1}^n v_i x_i \right)$. Recall we define $S$ as the set of high value items -- $S := \{ i \mid v_i \geq \frac{OPT}{t - 1} \}$

First, we use \Cref{red-x-i} to deduce
\begin{align*}
\pex{D}\left( \sum_{i = 1}^n v_i x_i \right) &= \sum_{i = 1}^n v_i  \pex{D}\left(x_i \right) \\
&= \sum_{i = 1}^n v_i \pex{D}\left( \sum_{U \subseteq S} \rd(f_{S, U}) x_i \right) \\
&= \sum_{i = 1}^n \sum_{U \subseteq S} v_i \pex{D}(\rd(f_{S, U}) x_i) \\
&= \sum_{U \subseteq S} \pex{D} \left( \sum_{i = 1}^n v_i x_i \rd(f_{S, U})\right)
\end{align*}
Next, we partition the inner sum into terms in $S$ and terms not in $S$ --  
\begin{align*}
    \pex{D}\left( \sum_{i = 1}^n v_i x_i \right) &= \sum_{U \subseteq S} \pex{D} \left( \sum_{i = 1}^n v_i x_i \rd(f_{S, U})\right) \\
    &= \sum_{U \subseteq S} \pex{D} \left( \sum_{i \in S} v_i x_i \rd(f_{S, U}) + \sum_{i \not \in S} v_i x_i \rd(f_{S, U}) \right) \\
    &= \sum_{U \subseteq S} \pex{D} \left( \sum_{i \in S} v_i x_i \rd(f_{S, U}) + \sum_{i \not \in S} v_i x_i \rd(f_{S, U})\right) \\
    &= \sum_{U \subseteq S} \pex{D} \left( \sum_{i \in S} v_i x_i \rd(f_{S, U}) \right) + \sum_{U \subseteq S} \pex{D} \left( \sum_{i \not \in S} v_i x_i \rd(f_{S, U}) \right).
\end{align*}
For the first term, we note that, restricted to $S$, $D$ has a corresponding global distribution which agrees with $D$ for low-degree terms. Also, note that $\rd{f_{S, U}}$ takes only takes on values $0$ and 1. Hence, for a given $U \subseteq S$, 
\[ \pex{D} \left( \sum_{i \in S} v_i x_i \rd(f_{S, U}) \right) = \left( \sum_{i \in U} v_i \right) \pex{D}(\rd(f_{S, U})). \]
For the latter term, we use the definition of $y_i^U$. Recall, 
\[ \pex{D}(x_i \cdot f_{S, U}) =  y_i^U \cdot \pex{D}(f_{S, U}). \]
Hence, for a given $U \subseteq S$, we have 
\[ \pex{D} \left( \sum_{i \not \in S} v_i x_i \rd(f_{S, U}) \right) = \pex{D}(\rd(f_{S, U}) \cdot \sum_{i \not \in S} v_i y_i^U. \]
In total, we have 
\begin{align*}
    \pex{D}\left( \sum_{i = 1}^n v_i x_i \right) &= \sum_{U \subseteq S} \pex{D} \left( \sum_{i \in S} v_i x_i \rd(f_{S, U}) \right) + \sum_{U \subseteq S} \pex{D} \left( \sum_{i \not \in S} v_i x_i \rd(f_{S, U}) \right) \\
    &= \sum_{U \subseteq S} \left( \sum_{i \in U} v_i \right) \pex{D}(\rd(f_{S, U})) + \sum_{U \subseteq S} \pex{D}(\rd(f_{S, U}) \cdot \sum_{i \not \in S} v_i y_i^U.
\end{align*}

At this point, we note that the term $\sum_{i \not \in S} v_i y_i^{U}$ can be upper bounded by the linear program 
\begin{align*}
    \max \sum_{i \not \in S} y_i v_i  \\
    \sum_{i \not \in S} y_i c_i \leq C \\
    0 \leq y_i \leq 1
\end{align*}
which we know will have value at most $OPT'$ + $OPT/t$ where $OPT'$ the the optimal value given capacity $C - \sum_{i \in U} c_i$. So, 
\begin{align*}
    \pex{D}\left( \sum_{i = 1}^n v_i x_i \right) &= \sum_{U \subseteq S} \left( \sum_{i \in U} v_i \right) \pex{D}(\rd(f_{S, U})) + \sum_{U \subseteq S} \pex{D}(\rd(f_{S, U}) \cdot \sum_{i \not \in S} v_i y_i^U \\
    &\leq \sum_{U \subseteq S} \pxd(\rd(f_{S, U}))\left( \sum_{i \in U} v_i + OPT' + OPT/t \right) \\
    &\leq \left(1 + \frac{1}{t} \right) OPT \sum_{U \subseteq S} \pxd(\rd(f_{S, U}) \\
    &= \left(1 + \frac{1}{t} \right) OPT.
\end{align*}

\end{proof}

\subsection{The Lemmas}

We use the following fact in our analysis: 
\begin{lemma}\label{fact}
For all $U \subseteq S$, 
\[ \pex{D}(\rd(f_{S, U}^2 x_i)) = \pex{D}(\rd(f_{S, U})^2) \]
\end{lemma}
\begin{proof}
For a given $U \subseteq S$ and $i \in [n]$, we have 
\[ f_{S, U}^2  = f_{S, U} . \]
since $f$ is an indicator function taking on values 0 and 1. Applying $\rd$ to both sides, we get
\[ \rd(f_{S, U}^2) = \rd(f_{S, U}) \]
which implies 
\[ \rd(f_{S, U}^2)x_i = \rd(f_{S, U}) x_i. \]
To reason about $\pex{D}(\rd(f_{S, U})^2)$, we note
\[ \rd(f_{S, U})^2 = \rd(f_{S, U}^2) + g \]
where $g$ is a polynomial of degree at least $t$. We know that $\pex{D}(X_T \cdot x_i) = 0$ for all $T \subseteq S$ of size at least $t$. Thus, $\pex{D}(g \cdot x_i) = 0$. This proves our claim. 
\end{proof}
We proceed to prove \Cref{y}. 
\begin{proof}[Proof of \Cref{y}]
First, we prove that $y_i^U \leq 1$. Recall 
\[ y_i^U = \frac{\pex{D}(x_i \cdot \rd(f_{S, U}))}{\pex{D}(\rd(f_{S, U})}. \]
Our first case will be if $\pxd(\rd(f_{S, U})) = 0$. In this case, by Cauchy-Schwartz (note that $\rd{f_{S, U}}$ and $x_i$ are both polynomials of degree at most $t$), we have  
\begin{align*}
    \pex{D}(x_i \cdot \rd(f_{S, U}))) &\leq \sqrt{\pex{D}(\rd(f_{S, U})^2)} \sqrt{\pex{D}(x_i^2)} \\
    &= \sqrt{\pex{D}(\rd(f_{S, U}))} \cdot \sqrt{\pex{D}(x_i^2)}  \\
    &= 0 \cdot \sqrt{\pex{D}(x_i^2)} \\
    &= 0.
\end{align*}
Hence, we consider the case where $\pex{D}(\rd(f_{S, U})$ is not zero. Here, we have 
\[ y_i^U = \frac{\pex{D}(x_i \cdot \rd(f_{S, U}))}{\pex{D}(\rd(f_{S, U}))} = \frac{\pex{D}(x_i \cdot \rd(f_{S, U})^2)}{\pex{D}(\rd(f_{S, U}))}. \]
Note that since $(1 - x_i) = (1 - x_i)^2$ (since $x_i$ takes on values 0 or 1), we have 
\begin{align*}
     \pxd(\rd(f_{S, U})^2(1 - x_i)) &= \pxd(\rd(f_{S, U})^2(1 - 2x_i + x_i^2)) \\
     &= \pxd(\rd(f_{S, U})^2(1 - x_i)^2) \\
     &\geq 0.
\end{align*}
Therefore, $\pex{D}(x_i \cdot \rd(f_{S, U})) = \pex{D}(x_i \cdot \rd(f_{S, U})^2) \leq \pex{D}(\rd(f_{S, U}))$, which implies $y_i^U \leq 1$. 
\[\]
Next, we prove $y_i^U \geq 0$. The denominator satisfies 
\[ \pxd(\rd(f_{S, U})) =  \pxd(\rd(f_{S, U})^2) \geq 0\]
and the numerator satisfies (since it is a product of indicator functions) 
\[ \pxd(\rd(f_{S, U}) x_i ) =  \pxd(\rd(f_{S, U})^2 x_i^2) \geq 0.\]
\[\]
Lastly, it remains to prove 
 \[ \sum_{i \not \in S}y_i^U c_i \leq C - \sum_{i \in U} c_i.\]
 Note that by the definition of $\rd(f_{S, U})$ and Remark \ref{looks-like-mu} ,
\begin{align}
    \pxd\left( \sum_{i \in S} x_i c_i \rd(f_{S, U})\right) &= \pex{\mu} \left( \sum_{i \in S } x_i c_i f_{S, U} \right) \\
    &= \sum_{i \in U} c_i \pxd(\rd(f_{S, U})).
\end{align}
On the other hand, 
\[ \pxd\left( \sum_{i \in [n]} x_i c_i \rd(f_{S, U}) \right) = \pxd\left( \sum_{i \in [n]} x_i c_i \rd(f_{S, U})^2 \right). \]
Since $\rd(f_{S, U})^2$ is a square of a degree $t - 1$ polynomial and $D$ is a pseudo-distribution satisfying the appropriate constraints, we have 
\[ \pxd\left(\rd(f_{S, U})^2 \left(C - \sum_{i \in [n]} c_i x_i \right)\right) \geq 0. \]
In particular, we have 
\begin{align}
    C \cdot \pxd(\rd(f_{S, U})) = C \cdot \pxd(\rd(f_{S, U})^2) \geq \pxd \left(\sum_{i \in [n]} c_i x_i \rd(f_{S, U}) \right).
\end{align} 
Subtracting (1) from (3), we get
\[ \pxd\left( \sum_{i \not \in S} \rd(f_{S, U}) x_i c_i \right) \leq C \cdot \pxd(f_{S, U}) - \sum_{i \in U} c_i \pxd(\rd(f_{S, U})). \]
This gives us our desired result.
\end{proof}

\bibliographystyle{alpha}
\bibliography{refs}

\end{document}